\documentclass[preprint]{elsarticle}

\pdfoutput=1

\usepackage{amssymb}
\usepackage{amsmath}
\usepackage{amstext}
\usepackage{amsthm}
\usepackage{color,hcolor}
\usepackage{graphicx}
\usepackage[colorlinks, citecolor=blue, linkcolor=blue,
urlcolor=blue]{hyperref}
\usepackage{amsfonts}\usepackage{mathrsfs}
\usepackage[normalem]{ulem}
\usepackage{mathptmx}

\DeclareSymbolFont{cmsy}{OMS}{cmsy}{m}{n}
\DeclareSymbolFontAlphabet{\mathcalus}{cmsy}

\newtheorem{theorem}{Theorem}[]
\newtheorem{proposition}[theorem]{Proposition}

\newcommand{\mrv}[1]{\mathbf{#1}}

\newcommand{\transp}[0]{\mathrm{T}}

\newcommand{\rel}[0]{\mathrm{rel}}
\newcommand{\secur}[0]{\mathrm{sec}}

\newcommand{\tr}{\operatorname{tr}}
\newcommand{\ket}[1]{\lvert{#1} \rangle}

\newcommand{\bra}[1]{{\langle {#1}\rvert}}

\allowdisplaybreaks
\newcommand{\bi}{\mathbf{i}}

\newcommand{\bc}{\mathbf{c}}
\newcommand{\bk}{\mathbf{k}}

\newcommand{\bb}{\mathbf{b}}
\newcommand{\bx}{\mathbf{x}}

\newcommand{\bz}{\mathbf{z}}
\newcommand{\bT}{\mathbf{T}}

\newcommand{\bxi}{{\boldsymbol{\xi}}}

\newcommand{\bs}{\mathbf{s}}
\newcommand{\Eve}{{\operatorname{Eve}}}

\newcommand{\Span}{{\operatorname{Span}}}

\begin{document}

\begin{frontmatter}

\title{Composable security against collective attacks of
a modified BB84 QKD protocol with information only
in one basis\tnoteref{prelimver}}

\tnotetext[prelimver]{A preliminary version of this paper appeared in
\emph{Proceedings of the 2nd International Conference on Complexity,
Future Information Systems and Risk -- COMPLEXIS, 24-26 April, 2017,
Porto, Portugal}~\cite{BLM2017_bb84_info_z}.}

\author[montreal]{Michel Boyer}
\ead{boyer@iro.umontreal.ca}
\address[montreal]{D\'epartement IRO, Universit\'e de Montr\'eal,
Montr\'eal (Qu\'ebec) H3C 3J7, Canada}
\author[technion]{Rotem Liss\corref{cor}}
\ead{rotemliss@cs.technion.ac.il}
\cortext[cor]{Corresponding author}
\author[technion]{Tal Mor}
\ead{talmo@cs.technion.ac.il}
\address[technion]{Computer Science Department,
Technion, Haifa 3200003, Israel}

\begin{abstract}
Quantum Cryptography uses the counter-intuitive properties
of Quantum Mechanics for performing cryptographic tasks in
a secure and reliable way.
The Quantum Key Distribution (QKD) protocol BB84 has been proven secure
against several important types of attacks:
collective attacks and joint attacks.
Here we analyze the security of a modified BB84
protocol, for which information is sent only in the $z$ basis
while testing is done in both the $z$ and the $x$ bases,
against collective attacks.
The proof follows the framework of a previous paper~\cite{BGM09},
but it avoids a classical information-theoretical analysis
and proves a fully composable security.
We show that this modified BB84 protocol is as
secure against collective attacks as the original BB84 protocol,
and that it requires more bits for testing.
\end{abstract}

\begin{keyword}
Collective Attacks\sep Quantum Key Distribution\sep Cryptography\sep
Error Rate\sep Test Bits\sep Information Bits.
\end{keyword}

\end{frontmatter}

\section{Introduction}
Cryptography is the science of protecting the security and correctness
of data against adversaries. One of the most important cryptographic
problems is the problem of \emph{encryption} -- namely, of transmitting
a secret message from a sender to a receiver.
Two main encryption methods are used today:
\begin{enumerate}
\item In symmetric-key cryptography, the same \emph{secret key}
is used for both the sender and the receiver:
the sender uses the secret key for encrypting his or her message,
and the receiver uses the same secret key for decrypting the message.
Examples of symmetric-key ciphers include
the Advanced Encryption Standard (AES)~\cite{AESbook},
the older Data Encryption Standard (DES),
and one-time pad (``Vernam cipher'').
\item In public-key cryptography~\cite{diffie_hellman76},
a \emph{public key} (known to everyone) and a \emph{secret key}
(known only to the receiver) are used:
the sender uses the public key for encrypting his or her message,
and the receiver uses the secret key for decrypting the message.
Examples of public-key ciphers include RSA~\cite{rsa} and
elliptic curve cryptography.
\end{enumerate}
One of the main problems with current public-key cryptography is that
its security is usually not formally proved.
Moreover, its security relies on the computational hardness
of specific computational problems,
such as integer factorization and discrete logarithm
(that can both be efficiently solved on a quantum computer,
by using Shor's factorization algorithm~\cite{shor94};
therefore, if a scalable quantum computer is successfully built
in the future, the security of many public-key ciphers,
including RSA and elliptic curve cryptography, will be broken).
Symmetric-key cryptography requires a secret key to be shared
\emph{in advance} between the sender and the receiver (in other words,
if the sender and the receiver want to share a secret message,
they must share a secret key beforehand). Moreover, no security proofs
for many current symmetric-key ciphers, such as AES and DES, are known
(even if one is allowed to rely on
the computational hardness of problems),
and unconditional security proofs
against computationally-unlimited adversaries are
impossible unless the secret key is used only once and
is at least as long as the secret message~\cite{shannon_secrecy49}.

The one-time pad (symmetric-key) cipher, that, given a message $M$
and a secret key $K$ of the same length,
defines the encrypted message $C$ to be $C = M \oplus K$
(and then decryption can be performed by computing $M = C \oplus K$),
is fully and unconditionally secure
against any adversary~\cite{shannon_secrecy49}:
namely, even if the adversary Eve intercepts the encrypted message $C$,
she gains no information about the original message $M$
(assuming that she has no information about the secret key $K$).
This means that, for obtaining perfect secrecy, all that is needed
is an efficient way for sharing a random secret key between the sender
and the receiver; unfortunately, ``classical key distribution''
cannot be achieved in a fully secure way if the adversary can listen
to all the communication between the sender and the receiver.

Quantum key distribution (QKD) protocols take advantage of the laws
of quantum mechanics for achieving fully and unconditionally secure
key distribution, so that their resulting final key can
later be used by other cryptographic primitives
(e.g., one-time pad encryption).
Most of the QKD protocols have security proofs applicable
even against adversaries whose only limitations are the laws of nature
(and who are otherwise capable of solving any computational problem
and of performing any physically-allowed operation).
The two parties (the first party is usually named ``Alice'',
and the second party is usually named ``Bob'')
want to create a shared random key, and they use
an insecure quantum channel and an unjammable classical channel
(to which the adversary may listen, but not interfere).
The adversary (eavesdropper), Eve, tries to get as much information
as she can on the final shared key.
The first and most important QKD protocol is BB84~\cite{BB84},
that uses four possible quantum states (see details below),
and it has been proven fully and unconditionally secure.

Boyer, Gelles, and Mor~\cite{BGM09} discussed the security of
BB84 against collective attacks.
The class of the ``collective attacks''~\cite{BM97a,BM97b,BBBGM02}
is an important and powerful subclass of the joint attacks;
the class of the ``joint attacks'' includes all the theoretical attacks
allowed by quantum physics.
\cite{BGM09} improved the security proof of Biham, Boyer, Brassard,
van de Graaf, and Mor~\cite{BBBGM02}
against collective attacks,
by using some techniques of Biham, Boyer, Boykin, Mor,
and Roychowdhury~\cite{BBBMR06}
(that proved security against joint attacks).
In this paper, too, we restrict the analysis to collective attacks,
because security against collective attacks is conjectured
(and, in some security notions, proved~\cite{renner_thesis08,CKR09})
to imply security against joint attacks.
In addition, proving security against collective attacks is
much simpler than proving security against joint attacks.

Other QKD protocols, either similar to BB84 or ones that use
different approaches, have also been suggested,
and in some cases have also been proven fully secure.
In particular, the ``three-state protocol''~\cite{Mor98}
uses only three quantum states,
and it has been proven secure~\cite{mor98_sec1,mor98_sec2,mor98_sec3};
the ``classical Bob'' protocol~\cite{cbob07} is a two-way protocol
such that only Alice has quantum capabilities
and Bob has only classical
capabilities, and it has been proven robust~\cite{cbob07}
and secure~\cite{cbob_security15};
and the ``classical Alice'' protocol~\cite{calice09}
is similar to ``classical Bob'' with Alice being
the classical participant instead of Bob,
and it has been proven robust~\cite{calice09comment}.

The above QKD protocols are all ``Discrete-Variable'' protocols.
Two other classes of QKD protocols, ``Continuous-Variable'' protocols
and ``Distributed-Phase-Reference'' protocols,
have also been suggested;
their security proofs are still weaker than the security proofs
of ``Discrete-Variable'' protocols
(see~\cite{sec_practical09} for details).

QKD protocols can be used as a subroutine (secure key distribution)
of more complicated cryptographic protocols. In other words, they can
be integrated into a system in order to improve its security.
See~\cite{SML10} for more details about this integration
and about the practical usability of QKD compared to other methods.

In many QKD protocols, including BB84, Alice and Bob exchange
several types of bits (encoded as quantum systems, usually qubits):
INFO bits, that are secret bits shared by Alice and Bob and are used
for generating the final key (via classical processes of
error correction and privacy amplification);
and TEST bits, that are publicly exposed by Alice and Bob
(by using the classical channel) and are used for
estimating the error rate. In BB84, each bit is sent from Alice
to Bob in a random basis (the $z$ basis or the $x$ basis).

In this paper, we extend the analysis of BB84 done in~\cite{BGM09}
and prove the security of a QKD protocol
we shall name \emph{BB84-INFO-$z$}.
This protocol is almost identical to BB84, except that all its INFO
bits are in the $z$ basis.
In other words, the $x$ basis is used only for testing. The bits
are thus partitioned into three disjoint sets: INFO, TEST-Z,
and TEST-X. The sizes of these sets are arbitrary
($n$ INFO bits, $n_z$ TEST-Z bits, and $n_x$ TEST-X bits).

We note that, while this paper follows a line of research that mainly
discusses a specific approach of security proof for BB84
and similar protocols (this approach,
notably, considers finite-key effects
and not only the asymptotic error rate),
many other approaches have also been suggested: see for example
\cite{bb84_sec_mayers,bb84_sec_SP,renner_thesis08,bb84_sec_renner}.
For comparison, see Section~\ref{conclusion}.

In the other papers (\cite{BM97a,BM97b,BBBGM02,BBBMR06,BGM09})
that discussed the same approach of security proofs as discussed here,
the classical mutual information between Eve and the final key
was calculated and bounded, which caused problems with composability
(see definition in~\cite{renner_thesis08}
and in Subsection~\ref{qkd_security_definitions}).
In contrast to those papers, in this paper we suggest a method
to prove a fully composable security -- namely, to calculate and bound
the trace distance between the quantum state at the end of the real
protocol and the quantum state at the end of an ideal protocol.
This method is fully composable, because it bounds the distance
between \emph{quantum} states instead of bounding
the \emph{classical} information Eve has
(bounding the classical information means,
in particular, that we assume that Eve measures at the end of
the protocol, while in reality she is not required to measure then,
but is allowed to wait until Alice and Bob use the final key).
This method is implemented in this paper for proving the fully
composable security of BB84-INFO-$z$ against collective attacks;
it also directly applies to the BB84 security proof
in~\cite{BGM09} against collective attacks, proving
the fully composable security of BB84 against collective attacks.
It may be extended in the future to show that the BB84
security proof of~\cite{BBBMR06} proves
the fully composable security of BB84 against joint attacks.
(We note that in the conference version of this
paper~\cite{BLM2017_bb84_info_z}, we used a weaker security definition:
it was not sufficient for proving fully composable security,
but it was more composable than in the previous papers.)

The ``qubit space'', $\mathcal{H}_2$, is a $2$-dimensional
Hilbert space. The states $\ket{0^0}, \ket{1^0}$ form an orthonormal
basis of $\mathcal{H}_2$, called ``the computational basis'' or
``the $z$ basis''. The states
$\ket{0^1} \triangleq \frac{\ket{0^0} + \ket{1^0}}{\sqrt{2}}$ and
$\ket{1^1} \triangleq \frac{\ket{0^0} - \ket{1^0}}{\sqrt{2}}$
form another orthonormal basis of $\mathcal{H}_2$,
called ``the $x$ basis''. Those two bases are said
to be \emph{conjugate bases}.

In this paper, we denote bit strings (of $t$ bits,
with $t \ge 0$ being some integer) by a bold letter
(e.g., $\bi=i_1\ldots i_t$ with $i_1,\ldots,i_t \in \{0,1\}$);
and we refer to those bit strings as elements of $\mathbf{F}_2^t$
-- that is, as elements of a $t$-dimensional vector space over
the field $\mathbf{F}_2 = \{0,1\}$, where addition of two vectors
corresponds to a XOR operation between them.
The number of $1$-bits in a bit string $\bs$ is denoted by $|\bs|$,
and the Hamming distance between two strings $\bs$ and $\bs'$
is $d_H(\bs, \bs') = |\bs + \bs'|$.

\subsection{\label{qkd_security_definitions}Security Definitions of
Quantum Key Distribution}
Originally, a QKD protocol was defined to be secure if the
(classical) \emph{mutual information} between Eve's information
and the final key, maximized over all the possible attack strategies
and measurements by Eve, is exponentially small in the number of
qubits, $N$. Examples of security proofs of BB84 that use this
security definition are~\cite{bb84_sec_mayers,BBBMR06,bb84_sec_SP}.
Those security proofs used the observation that one cannot analyze
the \emph{classical} data held by Eve before privacy amplification
(as done in~\cite{BBCM95}), but must analyze the \emph{quantum} state
held by Eve~\cite{BMS96}. In other words, they assumed that Eve
could keep her quantum state until the end of the protocol,
and only \emph{then} choose the optimal measurement
(based on all the data she observed) and perform the measurement.

Later, it was noticed that this security definition may not be
``composable''. In other words, the final key is secure if Eve
measures the quantum state she holds at the end of the QKD protocol,
but the proof does not apply to \emph{cryptographic applications}
(e.g., encryption) of the final key:
Eve might gain non-negligible information after the key is used,
even though her information on the key itself was negligible.
This means that the proof is not sufficient for practical purposes.
In particular, those applications may be insecure if Eve keeps her
quantum state until Alice and Bob use the key (thus giving Eve some
new information) and only \emph{then} measures.

Therefore, a new notion of ``(composable) full security''
was defined~\cite{BHLMO05,bb84_sec_renner,renner_thesis08}
by using the trace distance between quantum states,
following universal composability definitions
for non-quantum cryptography~\cite{compos01_universal,compos00}.
Intuitively, this notion means that the final joint quantum state of
Alice, Bob, and Eve at the end of the protocol is \emph{very close}
(namely, the trace distance is exponentially small in $N$)
to their final state at the end of an \emph{ideal} key distribution
protocol, that distributes a \emph{completely random} and
\emph{secret} final key to both Alice and Bob.
In other words, if a QKD protocol is secure, then
except with an exponentially small probability,
one of the two following events happens:
the protocol is aborted, \emph{or} the secret key generated
by the protocol is the same as a perfect key
that is uniformly distributed
(i.e., each possible key having the same probability),
is the same for both parties,
and is independent of the adversary's information.

Formally, $\rho_{ABE}$ is defined as the final quantum state
of Alice, Bob, and Eve at the end of the protocol
(with Alice's and Bob's states being
simply the ``classical'' states $\ket{k_A}_A$ and $\ket{k_B}_B$,
where $k_A$ and $k_B$ are bit strings that are the final keys held
by Alice and Bob, respectively (note that usually $k_A = k_B$);
and with Eve's state including both her quantum probe and the
classical information published in the unjammable classical channel);
$\rho_U$ is defined as the complete mixture of all the possible keys
that are the same for Alice and Bob (namely, if the set of possible
final keys is $K$, then $\rho_U = \frac{1}{|K|}
\sum_{k \in K} \ket{k}_A \ket{k}_B \bra{k}_A \bra{k}_B$);
and $\rho_E$ is defined as the partial trace of $\rho_{ABE}$
over the system $AB$. For the QKD protocol
to be fully (and composably) secure, it is required that
\begin{equation}
\frac{1}{2} \tr \left| \rho_{ABE} - \rho_U \otimes \rho_E \right|
\le \epsilon, \label{definition_composable_security}
\end{equation}
where $\epsilon$ is exponentially small in $N$.
Intuitively, $\rho_{ABE}$ is the \emph{actual} joint state
of Alice, Bob, and Eve at the end of the QKD protocol;
$\rho_U$ is the \emph{ideal} final state of Alice and Bob
(an equal mixture of all the possible final keys,
that is completely uncorrelated with Eve
and is the same for Alice and Bob);
and $\rho_E$ is the state of Eve,
uncorrelated with the states of Alice and Bob.
Note that cases in which the protocol
is aborted are represented by the zero operator:
see~\cite[Subsection 6.1.2]{renner_thesis08} for details.

Composable security of many QKD protocols, including BB84,
has been proved~\cite{BHLMO05,bb84_sec_renner,renner_thesis08}.

\section{\label{bb84-info-z_description}Full Definition
of the ``BB84-INFO-$z$'' Protocol}
Below we formally define all the steps of the BB84-INFO-$z$ protocol,
as used in this paper.
\begin{enumerate}
\item Before the protocol, Alice and Bob choose
some shared (and public) parameters:
numbers $n$, $n_z$, and $n_x$ (we denote $N \triangleq n + n_z + n_x$),
error thresholds $p_{a,z}$ and $p_{a,x}$,
an $r \times n$ parity check matrix $P_C$
(corresponding to a linear error-correcting code $C$),
and an $m \times n$ privacy amplification matrix $P_K$
(representing a linear key-generation function).
It is required that \emph{all} the $r+m$ rows of the matrices
$P_C$ and $P_K$ put together are linearly independent.

\item Alice randomly chooses a partition
$\mathcalus{P} = (\bs, \bz, \bb)$ of the $N$ bits by randomly choosing
three $N$-bit strings $\bs, \bz, \bb \in \mathbf{F}^{N}_2$
that satisfy $|\bs| = n, |\bz| = n_z, |\bb| = n_x$, and
$|\bs + \bz + \bb| = N$. Thus, $\mathcalus{P}$ partitions the set
of indexes $\lbrace 1, 2, ..., N \rbrace$ into three disjoint sets:
\begin{itemize}
\item $I$ (INFO bits, where $s_j = 1$) of size $n$;
\item $T_Z$ (TEST-Z bits, where $z_j = 1$) of size $n_z$; and
\item $T_X$ (TEST-X bits, where $b_j = 1$) of size $n_x$.
\end{itemize}

\item Alice randomly chooses an $N$-bit string
$\bi \in \mathbf{F}^{N}_2$ and sends the $N$ qubit states
$\ket{i_1^{b_1}}, \ket{i_2^{b_2}}, \ldots, \ket{i_{N}^{b_{N}}}$,
one after the other, to Bob using the quantum channel.
Notice that Alice uses the $z$ basis for sending the INFO and TEST-Z bits,
and that she uses the $x$ basis for sending the TEST-X bits.
Bob keeps each received qubit in quantum memory,
not measuring it yet\footnote{
Here we assume that Bob has a quantum memory and can delay his
measurement. In practical implementations, Bob usually cannot do that,
but is assumed to measure in a randomly-chosen basis ($z$ or $x$),
so that Alice and Bob later discard the qubits measured in the wrong
basis. In that case, we need to assume that Alice sends more than
$N$ qubits, so that $N$ qubits are finally detected by Bob
and measured in the correct basis.
In the original scheme, the probability of choosing
each basis ($z$ or $x$) was $\frac{1}{2}$,
which caused half of the sent qubits to be lost;
in the improved scheme suggested by~\cite{LCA05},
the probability of choosing the $z$ basis can be much higher,
which means that much less qubits get lost.}.

\item Alice sends to Bob over the classical channel
the bit string $\bb = b_1 \ldots b_N$.
Bob measures each of the qubits he saved in the correct basis (namely,
when measuring the $i$-th qubit, he measures it in the $z$ basis
if $b_i = 0$, and he measures it in the $x$ basis if $b_i = 1$).

The bit string measured by Bob is denoted by $\bi^B$.
If there is no noise and no eavesdropping, then $\bi^B = \bi$.

\item Alice sends to Bob over the classical channel the bit string $\bs$.
The INFO bits (that will be used for creating the final key)
are the $n$ bits with $s_j = 1$,
while the TEST-Z and TEST-X bits (that will be used for testing)
are the $n_z + n_x$ bits with $s_j = 0$.
We denote the substrings of $\bi, \bb$ that correspond to the INFO bits
by $\bi_{\bs}$ and $\bb_{\bs}$, respectively.

\item Alice and Bob both publish the bit values they have for
all the TEST-Z and TEST-X bits, and they compare the bit values.
If more than $n_z \cdot p_{a,z}$ TEST-Z bits are different
between Alice and Bob \emph{or} more than $n_x \cdot p_{a,x}$
TEST-X bits are different between them,
they abort the protocol. We note that $p_{a,z}$ and $p_{a,x}$
(the pre-agreed error thresholds) are the maximal allowed
error rates on the TEST-Z and TEST-X bits, respectively --
namely, in each basis ($z$ and $x$) separately.

\item The values of the remaining $n$ bits
(the INFO bits, with $s_j = $1) are kept in secret by Alice and Bob.
The bit string of Alice is denoted $\bx = \bi_{\bs}$,
and the bit string of Bob is denoted $\bx^B$.

\item Alice sends to Bob the \emph{syndrome} of $\bx$
(with respect to the error-correcting code $C$
and to its corresponding parity check matrix $P_C$),
that consists of $r$ bits and is defined as
$\bxi = \bx P_C^\transp$.
By using $\bxi$, Bob corrects the errors in his $\bx^B$ string
(so that it is the same as $\bx$).

\item The final key consists of $m$ bits and is defined as
${\bf k} = \bx P_K^\transp$. Both Alice and Bob compute it.
\end{enumerate}

The protocol is defined similarly to BB84 (and to its description
in~\cite{BGM09}), except that it uses the generalized bit numbers
$n$, $n_z$, and $n_x$ (numbers of INFO, TEST-Z, and TEST-X bits,
respectively);
that it uses the partition $\mathcalus{P} = (\bs, \bz, \bb)$ for
dividing the $N$-bit string $\bi$ into three disjoint sets of indexes
($I$, $T_Z$, and $T_X$); and that it uses two separate thresholds
($p_{a,z}$ and $p_{a,x}$) instead of one ($p_a$).

\section{Proof of Security for the BB84-INFO-$z$ Protocol
Against Collective Attacks}
\subsection{The General Collective Attack of Eve}
Before the QKD protocol is performed (and, thus, independently of
$\bi$ and $\mathcalus{P}$), Eve chooses some collective
attack to perform. A \emph{collective attack} is bitwise:
it attacks each qubit separately,
by using a separate probe (ancillary state).
Each probe is attached by Eve to the quantum state,
and Eve saves it in a quantum memory.
Eve can keep her quantum probes indefinitely, even after the final key
is used by Alice and Bob; and she can perform,
at any time of her choice, an optimal measurement
of all her probes together, chosen based on all the information she has
at the time of the measurement (including the classical information
sent during the protocol, and including the information she acquires
when Alice and Bob use the key).

Given the $j$-th qubit $\ket{i_j^{b_j}}_{T_j}$ sent
from Alice to Bob ($1 \le j \le N$),
Eve attaches a probe state $\ket{0^E}_{E_j}$
and applies some unitary operator $U_j$ of her choice to the
compound system $\ket{0^E}_{E_j} \ket{i_j^{b_j}}_{T_j}$.
Then, Eve keeps to herself (in a quantum memory) the subsystem $E_j$,
which is her probe state; and sends to Bob the subsystem $T_j$,
which is the qubit sent from Alice to Bob
(which may have been modified by her attack $U_j$).

The most general collective attack $U_j$ of Eve on the $j$-th qubit,
represented in the orthonormal basis
$\{\ket{0^{b_j}}_{T_j}, \ket{1^{b_j}}_{T_j}\}$, is
\begin{eqnarray}
U_j \ket{0^E}_{E_j} \ket{0^{b_j}}_{T_j}
&=& \ket{E^{b_j}_{00}}_{E_j} \ket{0^{b_j}}_{T_j}
+ \ket{E^{b_j}_{01}}_{E_j} \ket{1^{b_j}}_{T_j} \\
U_j \ket{0^E}_{E_j} \ket{1^{b_j}}_{T_j}
&=& \ket{E^{b_j}_{10}}_{E_j} \ket{0^{b_j}}_{T_j}
+ \ket{E^{b_j}_{11}}_{E_j} \ket{1^{b_j}}_{T_j},
\end{eqnarray}
where $\ket{E^{b_j}_{00}}_{E_j}$, $\ket{E^{b_j}_{01}}_{E_j}$,
$\ket{E^{b_j}_{10}}_{E_j}$, and $\ket{E^{b_j}_{11}}_{E_j}$
are non-normalized states in Eve's probe system $E_j$
attached to the $j$-th qubit.

We thus notice that Eve can modify the original \emph{product state}
of the compound system, $\ket{0^E}_{E_j} \ket{i_j^{b_j}}_{T_j}$,
into an \emph{entangled state}
(e.g., $\ket{E^{b_j}_{00}}_{E_j} \ket{0^{b_j}}_{T_j}
+ \ket{E^{b_j}_{01}}_{E_j} \ket{1^{b_j}}_{T_j}$).
Eve's attack may thus cause Bob's state
to become entangled with her probe.
On the one hand, this may give Eve some information on Bob's state;
on the other hand, this causes disturbance that may be detected by Bob.
The security proof shows that the information obtained by Eve and
the disturbance caused by Eve are inherently correlated:
this is the basic reason QKD protocols are secure.

\subsection{Results from~\cite{BGM09}}
The security proof of BB84-INFO-$z$ against collective attacks is very
similar to the security proof of BB84 itself against collective
attacks, that was detailed in~\cite{BGM09}. Most parts of the proof are
not affected at all by the changes made to BB84 to get the
BB84-INFO-$z$ protocol (changes detailed in
Section~\ref{bb84-info-z_description} of the current paper),
because those parts assume fixed strings $\bs$ and $\bb$,
and because the attack is collective (so the analysis is
restricted to the INFO bits).

Therefore, the reader is referred to the proof in Section~2 and
Subsections~3.1 to~3.5 of~\cite{BGM09}, that applies to BB84-INFO-$z$
without any changes (except changing the total number of bits,
$2n$, to $N$, which does not affect the proof at all),
and that will not be repeated here.

We denote the rows of the error-correction parity check matrix $P_C$
as the vectors $v_1,\ldots,v_r$ in $\mathbf{F}_2^n$,
and the rows of the privacy amplification matrix $P_K$
as the vectors $v_{r+1}, \ldots,v_{r+m}$.
We also define, for every $r'$,
$V_{r'} \triangleq \Span \lbrace v_1, ..., v_{r'} \rbrace$;
and we define
\begin{equation}
d_{r,m} \triangleq \min_{r
\leq r' < r+m} d_H(v_{r'+1}, V_{r'}) = \min_{r \leq r' < r+m}
d_{r',1}.
\end{equation}

For a $1$-bit final key $k \in \lbrace 0, 1 \rbrace$, we define
$\widehat{\rho}_k$ to be the state of Eve corresponding to the final
key $k$, given that she knows $\bxi$. Thus,
\begin{align}\label{rhohatk}
\widehat{\rho}_k &= \frac{1}{2^{n-r-1}}\sum_{\bx\, \big|
{\scriptsize\begin{matrix}\bx P_C^\transp = \bxi\\
\bx \cdot v_{r+1} = k\end{matrix}}} \rho_\bx^{\bb'},
\end{align}
where $\rho_\bx^{\bb'}$ is Eve's state after the attack, given that
Alice sent the INFO bit string $\bx$
encoded in the bases $\bb' = \bb_\bs$.
In~\cite{BGM09}, the state $\widetilde\rho_k$ was also defined:
it is a lift-up of $\widehat{\rho}_k$
(which means that $\widehat{\rho}_k$
is a partial trace of $\widetilde\rho_k$),
in which the states $\rho_\bx^{\bb'}$
appearing in $\widehat{\rho}_k$
are replaced by their purifications
(see full definition in Subsection~3.4 of~\cite{BGM09}).

In the end of Subsection~3.5 of~\cite{BGM09}, it was found that
(in the case of a $1$-bit final key, i.e., $m=1$)
\begin{equation}\label{boundrho}
\frac{1}{2} \tr |\widetilde\rho_0 - \widetilde\rho_1| \leq 2
\sqrt{P\left[|\mathbf{C}_{I}| \geq \frac{d_{r,1}}{2} \ \mid
\ \mathbf{B}_{I} =
\overline{\bb'}, \bs \right]},
\end{equation}
where $\mathbf{C}_{I}$ is a random variable whose value is
the $n$-bit string of errors on the $n$ INFO bits;
$\mathbf{B}_{I}$ is a random variable whose value is
the $n$-bit string of bases of the $n$ INFO bits;
$\overline{\bb'}$ is the bit-flipped string of $\bb' = \bb_\bs$;
and $d_{r,1}$ (and, in general, $d_{r,m}$) was defined above.

Now, according to \cite[Theorem~9.2 and page~407]{NCBook}, and using
the fact that $\widehat{\rho}_k$ is a partial trace of
$\widetilde\rho_k$, we find that
$\frac{1}{2} \tr |\widehat{\rho}_0 - \widehat{\rho}_1| \le
\frac{1}{2} \tr |\widetilde\rho_0 - \widetilde\rho_1|$. From this
result and from inequality \eqref{boundrho} we deduce that
\begin{equation}\label{boundrho2}
\frac{1}{2} \tr |\widehat{\rho}_0 - \widehat{\rho}_1| \leq 2
\sqrt{P\left[|\mathbf{C}_{I}| \geq \frac{d_{r,1}}{2} \ \mid
\ \mathbf{B}_{I} =
\overline{\bb'}, \bs \right]}.
\end{equation}

\subsection{Bounding the Differences Between Eve's States}
We define $\bc \triangleq \bi + \bi^B$: namely, $\bc$ is the XOR
of the $N$-bit string $\bi$ sent by Alice and
of the $N$-bit string $\bi^B$ measured by Bob.
For all indexes $1 \le l \le N$, $c_l=1$ if and only if
Bob's $l$-th bit value is different from the $l$-th bit sent by Alice.
The partition $\mathcalus{P}$ divides the $N$ bits into
$n$ INFO bits, $n_z$ TEST-Z bits, and $n_x$ TEST-X bits.
The corresponding substrings of the error string $\bc$ are
$\bc_\bs$ (the string of errors on the INFO bits),
$\bc_\bz$ (the string of errors on the TEST-Z bits), and
$\bc_\bb$ (the string of errors on the TEST-X bits).
The random variables that correspond to
$\bc_\bs$, $\bc_\bz$, and $\bc_\bb$ are denoted by
$\mrv{C}_I$, $\mrv{C}_{T_Z}$, and $\mrv{C}_{T_X}$, respectively.

We define $\widetilde{\mrv{C}_I}$ to be a random variable
whose value is the string of errors on the INFO bits
\emph{if Alice had encoded and sent the INFO bits in the $x$ basis}
(instead of the $z$ basis dictated by the protocol).
In those notations, inequality~\eqref{boundrho2} reads as
\begin{equation}
\frac{1}{2} \tr |\widehat{\rho}_0 - \widehat{\rho}_1| \leq 2
\sqrt{P\left[|\widetilde{\mathbf{C}_I}| \geq \frac{d_{r,1}}{2}
\ \mid\ \mathcalus{P} \right]} = 2
\sqrt{P\left[|\widetilde{\mathbf{C}_I}| \geq \frac{d_{r,1}}{2}
\ \mid\ \bc_\bz, \bc_\bb, \mathcalus{P} \right]},\label{sdrewritten}
\end{equation}
using the fact that Eve's attack is collective, so the qubits
are attacked independently, and, therefore, the errors on the INFO bits
are independent of the errors on the TEST-Z and TEST-X bits
(namely, of $\bc_\bz$ and $\bc_\bb$).

As explained in~\cite{BGM09}, inequality~\eqref{sdrewritten}
was not derived for the actual attack
$U = U_1\otimes \ldots \otimes U_N$ applied
by Eve, but for a virtual flat attack (that depends on $\bb$
and therefore could not have been applied by Eve).
That flat attack gives the same states $\widehat{\rho}_0$
and $\widehat{\rho}_1$ as given by the original attack $U$,
and it gives a lower (or the same) error rate in the conjugate basis.
Therefore, inequality~\eqref{sdrewritten} holds for
the original attack $U$, too.
This means that, starting from this point, all our results apply to
the original attack $U$ rather than to the flat attack.

So far, we have discussed a $1$-bit key.
We will now discuss a general $m$-bit key $\bk$.
We define $\widehat{\rho}_\bk$ to be the state of Eve corresponding
to the final key $\bk$, given that she knows $\bxi$:
\begin{align}\label{rhohatk_multi}
\widehat{\rho}_\bk &= \frac{1}{2^{n-r-m}}\sum_{\bx\,
\big|{\scriptsize\begin{matrix}\bx P_C^\transp =
\bxi\\ \bx P_K^\transp = \bk\end{matrix}}} \rho_\bx^{\bb'}
\end{align}

\begin{proposition}\label{probbound}
For any two keys $\bk,\bk'$ of $m$ bits,
\begin{equation}
\frac{1}{2} \tr |\widehat{\rho}_{\bk} - \widehat{\rho}_{\bk'}|
\leq 2m\sqrt{
P\left[| \widetilde{\mrv{C}_I} | \geq \frac{d_{r,m}}{2}
\mid \bc_\bz, \bc_\bb, \mathcalus{P}\right]}. \label{eqlemma}
\end{equation}
\end{proposition}
\begin{proof}
We define the key $\bk_j$, for $0 \le j \le m$, to consist of the first
$j$ bits of $\bk'$ and the last $m-j$ bits of $\bk$. This means that
$\bk_0 = \bk$, $\bk_m = \bk'$, and $\bk_{j-1}$ differs from $\bk_j$
at most on a single bit (the $j$-th bit).

First, we find a bound on $\frac{1}{2} \tr |\widehat{\rho}_{\bk_{j-1}}
- \widehat{\rho}_{\bk_j}|$: since $\bk_{j-1}$ differs from $\bk_j$
at most on a single bit (the $j$-th bit, given by the formula
$\bx \cdot v_{r+j}$), we can use the same proof
that gave us inequality~\eqref{sdrewritten}, attaching the other
(identical) key bits to $\bxi$ of the original proof; and we find that
\begin{equation}
\frac{1}{2} \tr |\widehat{\rho}_{\bk_{j-1}} - \widehat{\rho}_{\bk_j}|
\leq 2 \sqrt{P\left[|\widetilde{\mathbf{C}_I}| \geq \frac{d_j}{2}
\ \mid\ \bc_\bz, \bc_\bb, \mathcalus{P} \right]}, \label{StatesDiff1}
\end{equation}
where we define $d_j$ as $d_H(v_{r+j}, V'_j)$, and
$V'_j \triangleq \Span \lbrace v_1, v_2, \ldots, v_{r+j-1},
v_{r+j+1}, \ldots, v_{r+m} \rbrace$.

Now we notice that $d_j$ is the Hamming distance
between $v_{r+j}$ and some vector in $V'_j$,
which means that $d_j = \left|\sum_{i = 1}^{r+m} a_i v_i\right|$
with $a_i \in \mathbf{F}_2$ and $a_{r+j} \ne 0$.
The properties of Hamming distance assure us that $d_j$ is
at least $d_H(v_{r'+1}, V_{r'})$ for some $r \le r' < r + m$.
Therefore, we find that $d_{r,m} = \min_{r \leq r' < r+m}
d_H(v_{r'+1}, V_{r'}) \le d_j$.

The result $d_{r,m} \le d_j$ implies that
if $|\widetilde{\mathbf{C}_I}| \geq \frac{d_j}{2}$
then $|\widetilde{\mathbf{C}_I}|\geq \frac{d_{r,m}}{2}$.
Therefore, inequality~\eqref{StatesDiff1} implies
\begin{equation}
\frac{1}{2} \tr |\widehat{\rho}_{\bk_{j-1}} - \widehat{\rho}_{\bk_j}|
\leq 2 \sqrt{P\left[|\widetilde{\mathbf{C}_I}| \geq
\frac{d_{r,m}}{2} \ \mid \ \bc_\bz, \bc_\bb, \mathcalus{P} \right]}.
\label{StatesDiff2}
\end{equation}
Now we use the triangle inequality for norms to find
\begin{align}
\frac{1}{2} \tr |\widehat{\rho}_{\bk} - \widehat{\rho}_{\bk'}|
&= \frac{1}{2} \tr |\widehat{\rho}_{\bk_0} - \widehat{\rho}_{\bk_m}|
\leq \sum_{j = 1}^m \frac{1}{2} \tr |\widehat{\rho}_{\bk_{j-1}}
- \widehat{\rho}_{\bk_j}| \nonumber \\
&\leq 2m \sqrt{P\left[|\widetilde{\mathbf{C}_I}|
\geq \frac{d_{r,m}}{2}
\ \mid\ \bc_\bz, \bc_\bb, \mathcalus{P} \right]},\label{StatesDiff3}
\end{align}
as we wanted.
\end{proof}

We would now like to bound the expected value
(namely, the average value) of the trace distance between
two states of Eve corresponding to two final keys.
However, we should take into account that if the test fails,
no final key is generated,
in which case we define the distance to be $0$.
We thus define the random variable
$\Delta_\Eve^{(p_{a,z}, p_{a,x})}(\bk,\bk')$
for any two final keys $\bk,\bk'$:
\begin{equation}
\Delta_\Eve^{(p_{a,z}, p_{a,x})}(\bk, \bk' | \mathcalus{P},\bxi,
\bc_\bz, \bc_\bb)
\triangleq
\begin{cases} 
\frac{1}{2} \tr |\widehat{\rho}_{\bk} - \widehat{\rho}_{\bk'}| 
& \text{if $\displaystyle \frac{|\bc_\bz|}{n_z} \leq p_{a,z}$ and
$\displaystyle \frac{|\bc_\bb|}{n_x} \leq p_{a,x}$} \\
0 & \text{otherwise} \end{cases}\label{defipa}
\end{equation}

We need to bound the expected value
$\langle \Delta_\Eve^{(p_{a,z}, p_{a,x})}(\bk, \bk') \rangle$,
that is given by:
\begin{equation}
\langle \Delta_\Eve^{(p_{a,z}, p_{a,x})}(\bk, \bk')\rangle
= \sum_{\mathcalus{P},\bxi, \bc_\bz, \bc_\bb}
\Delta_\Eve^{(p_{a,z}, p_{a,x})}(\bk, \bk' | \mathcalus{P},\bxi,
\bc_\bz, \bc_\bb)
\cdot p(\mathcalus{P},\bxi, \bc_\bz, \bc_\bb)\label{eqipa}
\end{equation}
(In Subsection~\ref{sec_full_composability} we prove that this
expected value is indeed the quantity we need to bound
for proving fully composable security,
defined in Subsection~\ref{qkd_security_definitions}.)

\begin{theorem}\label{thmsecurity}
\begin{equation}
\langle \Delta_\Eve^{(p_{a,z}, p_{a,x})}(\bk, \bk')\rangle \leq 2m
\sqrt{P\left[\textstyle \left( \frac{| \widetilde{\mrv{C}_I}|}{n} \geq
\frac{d_{r,m}}{2n}\right)
\wedge \left(\frac{|\mrv{C}_{T_Z}|}{n_z} \leq
p_{a,z}\right)
\wedge \left(\frac{|\mrv{C}_{T_X}|}{n_x} \leq
p_{a,x}\right) \right]} \label{boundotherbasis}
\end{equation}
where $\frac {|\widetilde{\mrv{C}_I}|}{n}$ is a random variable
whose value is the error rate on the INFO bits
if they had been encoded in the $x$ basis,
$\frac{|\mrv{C}_{T_Z}|}{n_z}$ is a random variable
whose value is the error rate on the TEST-Z bits,
and $\frac{|\mrv{C}_{T_X}|}{n_x}$ is a random variable
whose value is the error rate on the TEST-X bits.
\end{theorem}
\begin{proof}
We use the convexity of $x^2$,
namely, the fact that for all $\{p_i\}_i$
satisfying $p_i \geq 0$ and $\sum_i p_i = 1$, it holds that
$\left( \sum_i p_i x_i \right)^2 \leq \sum_i p_i x_i^2$. We find that:
\begin{align}
&\langle \Delta_\Eve^{(p_{a,z}, p_{a,x})}(\bk, \bk')\rangle^2 &
\nonumber \\
&=
\left[ \sum_{\mathcalus{P},\bxi, \bc_\bz, \bc_\bb}
\Delta_\Eve^{(p_{a,z},
p_{a,x})}(\bk, \bk' | \mathcalus{P},\bxi, \bc_\bz, \bc_\bb)
\cdot p( \mathcalus{P}, \bxi, \bc_\bz, \bc_\bb)\right]^2
\qquad\qquad\qquad &\text{(by \eqref{eqipa})} \nonumber \\
&\leq \sum_{\mathcalus{P},\bxi, \bc_\bz, \bc_\bb} \left(
\Delta_\Eve^{(p_{a,z}, p_{a,x})}(\bk, \bk' |
\mathcalus{P},\bxi, \bc_\bz,
\bc_\bb) \right)^2
\cdot p(\mathcalus{P},\bxi, \bc_\bz, \bc_\bb)
\qquad\qquad\qquad &\text{\llap{(by convexity of $x^2$)}} \nonumber \\
&=
\sum_{\mathcalus{P},\bxi,\frac{|\bc_\bz|}{n_z}\leq p_{a,z},
\frac{|\bc_\bb|}{n_x}\leq p_{a,x}} 
\left( \textstyle \frac{1}{2} \tr |\widehat{\rho}_{\bk} -
\widehat{\rho}_{\bk'}| \right)^2
\cdot p( \mathcalus{P}, \bxi, \bc_\bz, \bc_\bb)
\qquad\qquad\qquad &\text{(by \eqref{defipa})} \nonumber \\
&\leq
4m^2 \cdot \sum_{\mathcalus{P},\bxi,\frac{|\bc_\bz|}{n_z}\leq p_{a,z},
\frac{|\bc_\bb|}{n_x}\leq p_{a,x}}
P\left[\textstyle | \widetilde{\mrv{C}_I}| \geq \frac{d_{r,m}}{2} \mid
\bc_\bz, \bc_\bb, \mathcalus{P}\right]
\cdot p(\mathcalus{P},\bxi, \bc_\bz, \bc_\bb)
\qquad &\text{(by \eqref{eqlemma})} \nonumber \\
&=
4m^2 \cdot \sum_{\mathcalus{P},\frac{|\bc_\bz|}{n_z}\leq p_{a,z},
\frac{|\bc_\bb|}{n_x}\leq p_{a,x}}
P\left[ \textstyle | \widetilde{\mrv{C}_I}| \geq \frac{d_{r,m}}{2}
\mid \bc_\bz, \bc_\bb, \mathcalus{P}\right]
\cdot p(\mathcalus{P},\bc_\bz, \bc_\bb) & \nonumber \\
&=
4m^2 \cdot \sum_{\mathcalus{P}}
P\left[ \textstyle \left(| \widetilde{\mrv{C}_I} | \geq
\frac{d_{r,m}}{2}\right)
\wedge \left( \frac{|\mrv{C}_{T_Z}|}{n_z} \leq
p_{a,z}\right)
\wedge \left( \frac{|\mrv{C}_{T_X}|}{n_x} \leq p_{a,x}\right) \mid
\mathcalus{P} \right] \cdot p(\mathcalus{P}) & \nonumber \\ 
&= 4m^2 \cdot P\left[ \textstyle \left(| \widetilde{\mrv{C}_I} | \geq
\frac{d_{r,m}}{2}\right)
\wedge \left( \frac{|\mrv{C}_{T_Z}|}{n_z} \leq
p_{a,z}\right)
\wedge \left( \frac{|\mrv{C}_{T_X}|}{n_x} \leq p_{a,x}\right)\right] &
\end{align}
as we wanted.
\end{proof}

\subsection{\label{sec_proof}Proof of Security}
Following~\cite{BGM09} and~\cite{BBBMR06},
we choose matrices $P_C$ and $P_K$ such that the inequality
$\frac{d_{r,m}}{2n} > p_{a,x} + \epsilon$
is satisfied for some $\epsilon$ (we will explain in
Subsection~\ref{sec_rel_rate_thres} why this is possible).
This means that
\begin{align}
& P\left[\textstyle \left( \frac{| \widetilde{\mrv{C}_I}|}{n} \geq
\frac{d_{r,m}}{2n}\right) \wedge
\left(\frac{|\mrv{C}_{T_Z}|}{n_z} \leq p_{a,z}\right) \wedge
\left(\frac{|\mrv{C}_{T_X}|}{n_x} \leq p_{a,x}\right) \right]
\nonumber \\
&\leq P\left[\textstyle \left(\frac{|\widetilde{\mrv{C}_I}|}{n} >
p_{a,x} + \epsilon\right) \wedge \left(\frac{|\mrv{C}_{T_X}|}{n_x} \leq
p_{a,x}\right) \right]. \label{PcPkInequality}
\end{align}
We will now prove the right-hand-side of~\eqref{PcPkInequality}
to be exponentially small in $n$.

As said earlier, the random variable $\widetilde{\mrv{C}_I}$
corresponds to the bit string of errors on the INFO bits
if they had been encoded in the $x$ basis. The TEST-X bits are
also encoded in the $x$ basis, and the random variable $\mrv{C}_{T_X}$
corresponds to the bit string of errors on those bits.
Therefore, we can treat the selection of the indexes
of the $n$ INFO bits and the $n_x$ TEST-X bits
as a random sampling
(after the numbers $n$, $n_z$, and $n_x$ \emph{and}
the indexes of the TEST-Z bits have all already been chosen)
and use Hoeffding's theorem
(that is described in Appendix~A of~\cite{BGM09}).

Therefore, for each bit string $c_1\ldots c_{n+n_x}$ that consists
of the errors in the $n+n_x$ INFO and TEST-X bits
\emph{if the INFO bits had been encoded in the $x$ basis},
we apply Hoeffding's theorem: namely, we take a sample of size $n$
without replacement from the population $c_1, \ldots, c_{n+n_x}$
(this corresponds to the random selection of
the indexes of the INFO bits and the TEST-X bits, as defined above,
given that the indexes of the TEST-Z bits have already been chosen).
Let $\overline{X} = \frac{|\widetilde{\mrv{C}_I}|}{n}$
be the average of the sample (this is exactly the error rate
on the INFO bits, assuming, again, that the INFO bits had been encoded
in the $x$ basis); and let
$\mu = \frac{|\widetilde{\mrv{C}_I}| + |\mrv{C}_{T_X}|}{n + n_x}$ 
be the expectancy of $\overline{X}$ (this is exactly the error rate
on the INFO bits and TEST-X bits together).
Then $\frac{|\mrv{C}_{T_X}|}{n_x} \leq p_{a,x}$ is equivalent to
$(n + n_x)\mu - n\overline{X} \leq n_x \cdot p_{a,x}$,
and, therefore, to $n \cdot (\overline{X}-\mu)
\geq n_x \cdot (\mu - p_{a,x})$.
This means that the conditions
$\left(\frac{|\widetilde{\mrv{C}_I}|}{n} > p_{a,x}+\epsilon\right)$ and
$\left(\frac{|\mrv{C}_{T_X}|}{n_x} \leq p_{a,x}\right)$ rewrite to
\begin{equation}
\left(\overline{X} - \mu > \epsilon + p_{a,x} - \mu\right)
\wedge \left(\frac{n}{n_x} \cdot (\overline{X}-\mu)
\geq \mu - p_{a,x}\right), \label{inequalities}
\end{equation}
which implies $\left(1 + \frac{n}{n_x}\right)(\overline{X} -\mu) >
\epsilon$, which is equivalent to $\overline{X} -\mu >
\frac{n_x}{n + n_x} \epsilon$.
Using Hoeffding's theorem (from Appendix~A of~\cite{BGM09}), we get:
\begin{equation}
P\left[ \left(\frac{|\widetilde{\mrv{C}_I}|}{n} > p_{a,x} +
\epsilon\right) \wedge \left(\frac{|\mrv{C}_{T_X}|}{n_x} \leq
p_{a,x}\right) \right]
\leq P\left[ \overline{X} - \mu > \frac{n_x}{n + n_x}\epsilon\right]
\leq e^{-2 \left(\textstyle\frac{n_x}{n + n_x}\right)^2 n\epsilon^2}
\end{equation}

In the above discussion, we have actually proved the following Theorem:
\begin{theorem}\label{thmsecurity1}
Let us be given $\delta > 0$, $R > 0$, and,
for infinitely many values of $n$,
a family $\{v^n_1, \ldots, v^n_{r_n+m_n}\}$ of
linearly independent vectors in $\mathbf{F}_2^n$ such that
$\delta < \frac{d_{r_n,m_n} }{n}$ and $\frac{m_n}{n} \leq R$.
Then for any $p_{a,z}, p_{a,x} > 0$ and $\epsilon_\secur > 0$
such that $p_{a,x} + \epsilon_\secur \leq \frac{\delta}{2}$,
and for any $n, n_z, n_x > 0$ and two $m_n$-bit final keys
$\bk,\bk'$, the distance between Eve's states
corresponding to $\bk$ and $\bk'$ satisfies the following bound:
\begin{equation}\label{EveInfoBound2}
\langle \Delta_\Eve^{(p_{a,z}, p_{a,x})}(\bk, \bk')\rangle \leq 2R\,
n e^{- \left(\textstyle\frac{n_x}{n + n_x}\right)^2 n \epsilon^2_\secur}
\end{equation}
\end{theorem}
In Subsection~\ref{sec_rel_rate_thres} we explain why
the vectors required by this Theorem exist.

We note that the quantity
$\langle \Delta_\Eve^{(p_{a,z}, p_{a,x})}(\bk, \bk')\rangle$ bounds
the expected values of the Shannon Distinguishability and of the
mutual information between Eve and the final key,
as done in~\cite{BGM09} and~\cite{BBBMR06},
which is sufficient for proving non-composable security;
but it also avoids composability problems:
Eve is not required to measure immediately after the protocol ends,
but she is allowed to wait until she gets more information.
In Subsection~\ref{sec_full_composability} we use this bound
for proving a fully composable security.

\subsection{\label{rel_proof}Reliability}
Security itself is not sufficient; we also need the key to be reliable
(namely, to be the same for Alice and Bob). This means that we should
make sure that the number of errors on the INFO bits is less than
the maximal number of errors that can be corrected by the
error-correcting code. We demand that our error-correcting code can
correct $n(p_{a,z}+\epsilon_\rel)$ errors (we explain in
Subsection~\ref{sec_rel_rate_thres} why this demand is satisfied).
Therefore, reliability of the final key with exponentially small
probability of failure is guaranteed by the following inequality:
(as said, $\mrv{C}_I$ corresponds to the actual bit string of errors
on the INFO bits in the protocol, when they are encoded
in the $z$ basis)
\begin{equation}
P\left[ \left(\frac{|\mrv{C}_I|}{n} > p_{a,z} + \epsilon_\rel\right)
\wedge \left(\frac{|\mrv{C}_{T_Z}|}{n_z} \leq p_{a,z}\right) \right]
\leq e^{-2 \left(\textstyle\frac{n_z}{n + n_z}\right)^2
n\epsilon_\text{rel}^2}
\end{equation}
This inequality is proved by an argument similar to the one used in
Subsection~\ref{sec_proof}: the selection of the indexes of the INFO
bits and the TEST-Z bits is a random partition of $n + n_z$ bits into
two subsets of sizes $n$ and $n_z$, respectively (assuming that the
indexes of the TEST-X bits have already been chosen),
and thus it corresponds to Hoeffding's sampling.

\subsection{\label{sec_full_composability}Proof of Fully Composable
Security}
We now prove that the BB84-INFO-$z$ protocol satisfies the definition
of composable security for a QKD protocol: namely, that it satisfies
equation~\eqref{definition_composable_security} presented in
Subsection~\ref{qkd_security_definitions}. We prove that the expression
$\frac{1}{2} \tr \left| \rho_{ABE} - \rho_U \otimes \rho_E \right|$
is exponentially small in $n$, with $\rho_{ABE}$ being the actual
joint state of Alice, Bob, and Eve;
$\rho_U$ being an ideal (random, secret, and shared)
key distributed to Alice and Bob;
and $\rho_E$ being the partial trace of $\rho_{ABE}$
over the system $AB$.

To make reading easier, we use the following notations,
where $\bi$ is the bit string sent by Alice,
$\bi^B$ is the bit string received by Bob,
and $\bc = \bi \oplus \bi^B$ is the string of errors:
\begin{eqnarray}
\bi^{AB}_\mathcalus{T} &\triangleq&
\left( \bi_\bz, \bi_\bb, \bi^B_\bz, \bi^B_\bb \right) \\
\bT &\triangleq& \begin{cases} 
1 & \text{if $\frac{|\bc_\bz|}{n_z} \leq p_{a,z}$ and
$\frac{|\bc_\bb|}{n_x} \leq p_{a,x}$} \\
0 & \text{otherwise} \end{cases}
\end{eqnarray}
In other words, $\bi^{AB}_\mathcalus{T}$ consists of all the TEST-Z
and TEST-X bits of Alice and Bob; and $\bT$ is the random variable
representing the result of the test.

According to the above definitions,
the states $\rho_{ABE}$ and $\rho_U$ are
\begin{eqnarray}
\rho_{ABE} &=& \sum_{\bi,\bi^B,\mathcalus{P} \mid \bT = 1}
P \left( \bi, \bi^B, \mathcalus{P} \right) \cdot
\ket{\bk}_A \bra{\bk}_A \otimes \ket{\bk'}_B \bra{\bk'}_B \nonumber \\
&\otimes& \left( \rho_{\bx,\bx^B}^{\bb'} \right)_E \otimes
\ket{\bi^{AB}_\mathcalus{T}, \mathcalus{P}, \bxi}_C
\bra{\bi^{AB}_\mathcalus{T}, \mathcalus{P}, \bxi}_C \\
\rho_U &=& \frac{1}{2^m} \sum_{\bk} \ket{\bk}_A \bra{\bk}_A \otimes
\ket{\bk}_B \bra{\bk}_B,
\end{eqnarray}
where $\left( \rho_{\bx,\bx^B}^{\bb'} \right)_E$
is defined to be Eve's quantum state if Alice sends
the INFO string $\bx = \bi_\bs$ in the bases $\bb' = \bb_\bs$
and Bob gets the INFO string $\bx^B = \bi^B_\bs$.
All the other states actually represent classical information:
subsystems $A$ and $B$ represent the final keys
held by Alice ($\bk = \bx P_K^\transp$) and Bob ($\bk'$,
that is obtained from $\bx^B$, $\bxi = \bx P_C^\transp$, and $P_K$),
and subsystem $C$ represents the information published
in the unjammable classical channel during the protocol
(this information is known to Alice, Bob, and Eve) --
namely, $\bi^{AB}_\mathcalus{T}$ (all the test bits),
$\mathcalus{P}$ (the partition),
and $\bxi = \bx P_C^\transp$ (the syndrome).

We note that in the definition of $\rho_{ABE}$, we sum only over
the events in which the test is \emph{passed}
(namely, in which the protocol is not aborted by Alice and Bob):
in such cases, an $m$-bit key is generated.
The cases in which the protocol aborts do not exist in the sum --
namely, they are represented by the zero operator,
as required by the definition of composable security
(see Subsection~\ref{qkd_security_definitions}
and~\cite[Subsection 6.1.2]{renner_thesis08}).
Thus, $\rho_{ABE}$ is a non-normalized state, and $\tr(\rho_{ABE})$
is the probability that the test is passed.

To help us bound the trace distance, we define
the following intermediate state:
\begin{eqnarray}
\rho_{ABE}' &\triangleq& \sum_{\bi,\bi^B,\mathcalus{P} \mid \bT = 1}
P \left( \bi, \bi^B, \mathcalus{P} \right) \cdot
\ket{\bk}_A \bra{\bk}_A \otimes \ket{\bk}_B \bra{\bk}_B \nonumber \\
&\otimes& \left( \rho_{\bx,\bx^B}^{\bb'} \right)_E \otimes
\ket{\bi^{AB}_\mathcalus{T}, \mathcalus{P}, \bxi}_C
\bra{\bi^{AB}_\mathcalus{T}, \mathcalus{P}, \bxi}_C
\end{eqnarray}
This state is identical to $\rho_{ABE}$, except that Bob holds
the Alice's final key ($\bk$) instead of his own calculated final key
($\bk'$). In particular, the similarity between
$\rho_{ABE}$ and $\rho_{ABE}'$ means, by definition, that
$\rho_E \triangleq \tr_{AB} \left( \rho_{ABE} \right)$
and $\rho_E' \triangleq \tr_{AB} \left( \rho_{ABE}' \right)$
are the same state: namely, $\rho_E = \rho_E'$.

\begin{proposition}
Under the same conditions as Theorem~\ref{thmsecurity1},
it holds that
\begin{eqnarray}
\frac{1}{2} \tr \left| \rho_{ABE}' - \rho_U \otimes \rho_E \right|
\le 2R\, n e^{- \left(\textstyle\frac{n_x}{n + n_x}\right)^2 n
\epsilon^2_\secur},
\end{eqnarray}
for $\rho_{ABE}'$ and $\rho_U$ defined above and for the partial trace
$\rho_E \triangleq \tr_{AB} \left( \rho_{ABE} \right)$.
\end{proposition}
\begin{proof}
We notice that in $\rho_{ABE}'$, the only factors depending directly
on $\bx$ and $\bx^B$ (and not only on $\bk$ and $\bxi$) are
the probability $P \left( \bi, \bi^B, \mathcalus{P} \right)$
and Eve's state $\left( \rho_{\bx,\bx^B}^{\bb'} \right)_E$.
The probability can be reformulated as:
\begin{eqnarray}
P \left( \bi, \bi^B, \mathcalus{P} \right) &=& P \left(
\bi^{AB}_\mathcalus{T}, \mathcalus{P}, \bxi \right) \cdot
P \left( \bk \mid \bi^{AB}_\mathcalus{T}, \mathcalus{P}, \bxi \right)
\nonumber \\
&\cdot& P \left( \bx \mid \bk, \bi^{AB}_\mathcalus{T}, \mathcalus{P}, \bxi
\right) \cdot P \left( \bx^B \mid \bx, \bk, \bi^{AB}_\mathcalus{T},
\mathcalus{P}, \bxi \right) \nonumber \\
&=& P \left( \bi^{AB}_\mathcalus{T}, \mathcalus{P}, \bxi \right)
\cdot \frac{1}{2^m} \cdot \frac{1}{2^{n-r-m}} \cdot
P \left( \bx^B \mid \bx, \bk, \bi^{AB}_\mathcalus{T},
\mathcalus{P}, \bxi \right)
\end{eqnarray}
(Because all the possible $n$-bit values of $\bx$
have the same probability, $\frac{1}{2^n}$; and because all the
$r+m$ rows of the matrices $P_C$ and $P_K$ are linearly independent,
so there are exactly $2^{n-r-m}$ values of $\bx$ corresponding to
each specific pair $(\bxi, \bk)$.)

Therefore, the state $\rho_{ABE}'$ takes the following form:
\begin{eqnarray}
\rho_{ABE}' &=& \frac{1}{2^m} \sum_{\bk,\bi^{AB}_\mathcalus{T},
\mathcalus{P},\bxi \mid \bT = 1}
P \left( \bi^{AB}_\mathcalus{T}, \mathcalus{P}, \bxi \right) \cdot
\ket{\bk}_A \bra{\bk}_A \otimes \ket{\bk}_B \bra{\bk}_B \nonumber \\
&\otimes& \left[ \frac{1}{2^{n-r-m}} \sum_{\bx,\bx^B
\big|{\scriptsize\begin{matrix}\bx P_C^\transp = \bxi\\
\bx P_K^\transp = \bk\end{matrix}}}
P \left( \bx^B \mid \bx, \bk, \bi^{AB}_\mathcalus{T}, \mathcalus{P},
\bxi \right)
\cdot \left( \rho_{\bx,\bx^B}^{\bb'} \right)_E \right] \nonumber \\
&\otimes& \ket{\bi^{AB}_\mathcalus{T}, \mathcalus{P}, \bxi}_C
\bra{\bi^{AB}_\mathcalus{T}, \mathcalus{P}, \bxi}_C \nonumber \\
&=& \frac{1}{2^m}
\sum_{\bk,\bi^{AB}_\mathcalus{T},\mathcalus{P},\bxi \mid \bT = 1}
P \left( \bi^{AB}_\mathcalus{T}, \mathcalus{P}, \bxi \right) \cdot
\ket{\bk}_A \bra{\bk}_A \otimes \ket{\bk}_B \bra{\bk}_B \nonumber \\
&\otimes& \left( \widehat{\rho}_\bk \right)_E
\otimes \ket{\bi^{AB}_\mathcalus{T}, \mathcalus{P}, \bxi}_C
\bra{\bi^{AB}_\mathcalus{T}, \mathcalus{P}, \bxi}_C
\end{eqnarray}
($\widehat{\rho}_\bk$ was defined in equation~\eqref{rhohatk_multi}.)

The partial trace $\rho_E' = \tr_{AB} \left( \rho_{ABE}' \right)$,
that (as proved above) is the same as $\rho_E$, is
\begin{equation}
\rho_E = \rho_E' = \frac{1}{2^m} \sum_{\bk,\bi^{AB}_\mathcalus{T},
\mathcalus{P},\bxi \mid \bT = 1}
P \left( \bi^{AB}_\mathcalus{T}, \mathcalus{P}, \bxi \right)
\cdot \left( \widehat{\rho}_\bk \right)_E
\otimes \ket{\bi^{AB}_\mathcalus{T}, \mathcalus{P}, \bxi}_C
\bra{\bi^{AB}_\mathcalus{T}, \mathcalus{P}, \bxi}_C,
\end{equation}
and the state $\rho_U \otimes \rho_E$ is
\begin{eqnarray}
\rho_U \otimes \rho_E &=& \frac{1}{2^{2m}}
\sum_{\bk,\bk'',\bi^{AB}_\mathcalus{T},\mathcalus{P},\bxi \mid \bT = 1}
P \left( \bi^{AB}_\mathcalus{T}, \mathcalus{P}, \bxi \right) \cdot
\ket{\bk}_A \bra{\bk}_A \otimes \ket{\bk}_B \bra{\bk}_B \nonumber \\
&\otimes& \left( \widehat{\rho}_{\bk''} \right)_E
\otimes \ket{\bi^{AB}_\mathcalus{T}, \mathcalus{P}, \bxi}_C
\bra{\bi^{AB}_\mathcalus{T}, \mathcalus{P}, \bxi}_C.
\end{eqnarray}
By using the triangle inequality for norms, since $\rho_{ABE}'$
and $\rho_U \otimes \rho_E$ are the same (except the difference
between Eve's states, $\left( \widehat{\rho}_{\bk} \right)_E$
and $\left( \widehat{\rho}_{\bk''} \right)_E$), we get, by using
the definition of $\langle \Delta_\Eve^{(p_{a,z}, p_{a,x})}(\bk, \bk'')
\rangle$ (equation~\eqref{eqipa}) and Theorem~\ref{thmsecurity1}:
\begin{eqnarray}
\frac{1}{2} \tr \left| \rho_{ABE}' - \rho_U \otimes \rho_E \right|
&\le& \frac{1}{2^{2m}} \sum_{\bk,\bk'',\bi^{AB}_\mathcalus{T},
\mathcalus{P},\bxi \mid \bT = 1}
P \left( \bi^{AB}_\mathcalus{T}, \mathcalus{P}, \bxi \right)
\cdot \frac{1}{2} \tr \left| \left( \widehat{\rho}_{\bk} \right)_E -
\left( \widehat{\rho}_{\bk''} \right)_E \right| \nonumber \\
&=& \frac{1}{2^{2m}} \sum_{\bk,\bk''}
\langle \Delta_\Eve^{(p_{a,z}, p_{a,x})}(\bk, \bk'') \rangle
\nonumber \\
&\le& 2R\, n e^{- \left(\textstyle\frac{n_x}{n + n_x}\right)^2 n
\epsilon^2_\secur}
\end{eqnarray}
as we wanted.
\end{proof}

We still have to bound the following difference:
\begin{eqnarray}
\rho_{ABE} - \rho_{ABE}' &=&
\sum_{\bi,\bi^B,\mathcalus{P} \mid \bT = 1}
P \left( \bi, \bi^B, \mathcalus{P} \right) \nonumber \\
&\cdot& \ket{\bk}_A \bra{\bk}_A
\otimes \left[ \ket{\bk'}_B \bra{\bk'}_B  -
\ket{\bk}_B \bra{\bk}_B \right] \nonumber \\
&\otimes& \left( \rho_{\bx,\bx^B}^{\bb'} \right)_E
\otimes \ket{\bi^{AB}_\mathcalus{T}, \mathcalus{P}, \bxi}_C
\bra{\bi^{AB}_\mathcalus{T}, \mathcalus{P}, \bxi}_C \nonumber \\
&=& P \left( \left( \bk \ne \bk' \right)
\wedge \left( \bT = 1 \right) \right) \nonumber \\
&\cdot& \sum_{\bi,\bi^B,\mathcalus{P}}
P \left( \bi, \bi^B, \mathcalus{P} \mid \left( \bk \ne \bk' \right)
\wedge \left( \bT = 1 \right) \right) \nonumber \\
&\cdot& \ket{\bk}_A \bra{\bk}_A
\otimes \left[ \ket{\bk'}_B \bra{\bk'}_B  -
\ket{\bk}_B \bra{\bk}_B \right] \nonumber \\
&\otimes& \left( \rho_{\bx,\bx^B}^{\bb'} \right)_E
\otimes \ket{\bi^{AB}_\mathcalus{T}, \mathcalus{P}, \bxi}_C
\bra{\bi^{AB}_\mathcalus{T}, \mathcalus{P}, \bxi}_C
\end{eqnarray}
Because the trace distance between every two normalized
states is bounded by $1$, and because of the reliability proof
in Subsection~\ref{rel_proof}, we get:
\begin{equation}
\frac{1}{2} \tr \left| \rho_{ABE} - \rho_{ABE}' \right|
\le P \left( \left( \bk \ne \bk' \right)
\wedge \left( \bT = 1 \right) \right)
\le e^{-2 \left(\textstyle\frac{n_z}{n + n_z}\right)^2
n\epsilon_\text{rel}^2}
\end{equation}
(Because if $\bk \ne \bk'$, Alice and Bob have different final keys,
and this means that the error correction stage did not succeed.
According to the discussion in Subsection~\ref{rel_proof},
this can happen only if there are too many errors
in the information string -- namely, if
$\frac{|\mrv{C}_I|}{n} > p_{a,z} + \epsilon_\rel$.)

To sum up, we get the following bound:
\begin{eqnarray}
\frac{1}{2} \tr \left| \rho_{ABE} - \rho_U \otimes \rho_E \right|
&\le& \frac{1}{2} \tr \left| \rho_{ABE} - \rho_{ABE}' \right|
+ \frac{1}{2} \tr \left| \rho_{ABE}' - \rho_U \otimes \rho_E \right|
\nonumber \\
&\le& e^{-2 \left(\textstyle\frac{n_z}{n + n_z}\right)^2
n\epsilon_\text{rel}^2}
+ 2R\, n e^{- \left(\textstyle\frac{n_x}{n + n_x}\right)^2 n
\epsilon^2_\secur}
\end{eqnarray}
This bound is exponentially small in $n$.
Thus, we have proved the composable security of BB84-INFO-$z$.

\subsection{\label{sec_rel_rate_thres}Security, Reliability,
and Error Rate Threshold}
According to Theorem~\ref{thmsecurity1} and to the discussion in
Subsection~\ref{rel_proof}, to get both security and reliability
we only need vectors $\{v^n_1, \ldots, v^n_{r_n+m_n}\}$
satisfying both the conditions of the Theorem (distance
$\frac{d_{r_n,m_n}}{2n} > \frac{\delta}{2}
\geq p_{a,x} + \epsilon_\secur$)
and the reliability condition
(the ability to correct $n(p_{a,z}+\epsilon_\rel)$ errors).
Such families were proven to exist in Appendix~E of~\cite{BBBMR06},
giving the following upper bound on the bit-rate:
\begin{eqnarray}
R_{\mathrm{secret}} \triangleq \frac{m}{n}
< 1 - H_2(2 p_{a,x} + 2 \epsilon_\secur)
- H_2 \left( p_{a,z} + \epsilon_\rel + \frac{1}{n} \right)
\end{eqnarray}
where $H_2(x) \triangleq -x \log_2(x) -(1-x) \log_2(1-x)$.

Note that we use here the error thresholds $p_{a,x}$ for security and
$p_{a,z}$ for reliability. This is possible,
because in~\cite{BBBMR06} those conditions (security and reliability)
on the codes are discussed separately.

To get the asymptotic error rate thresholds, we require
$R_{\mathrm{secret}} > 0$, and we get the condition:
\begin{equation}
H_2(2 p_{a,x} + 2 \epsilon_\secur) + H_2 \left( p_{a,z}
+ \epsilon_\rel + \frac{1}{n} \right) < 1
\end{equation}

The secure asymptotic error rate thresholds zone is shown in
Figure~\ref{fig:security} (it is below the curve),
assuming that $\frac{1}{n}$
is negligible. Note the trade-off between the error rates $p_{a,z}$ and
$p_{a,x}$. Also note that in the case $p_{a,z} = p_{a,x}$, we get the
same threshold as BB84 (\cite{BBBMR06} and~\cite{BGM09}), which is
7.56\%.

\begin{figure}
\includegraphics[width=\linewidth]{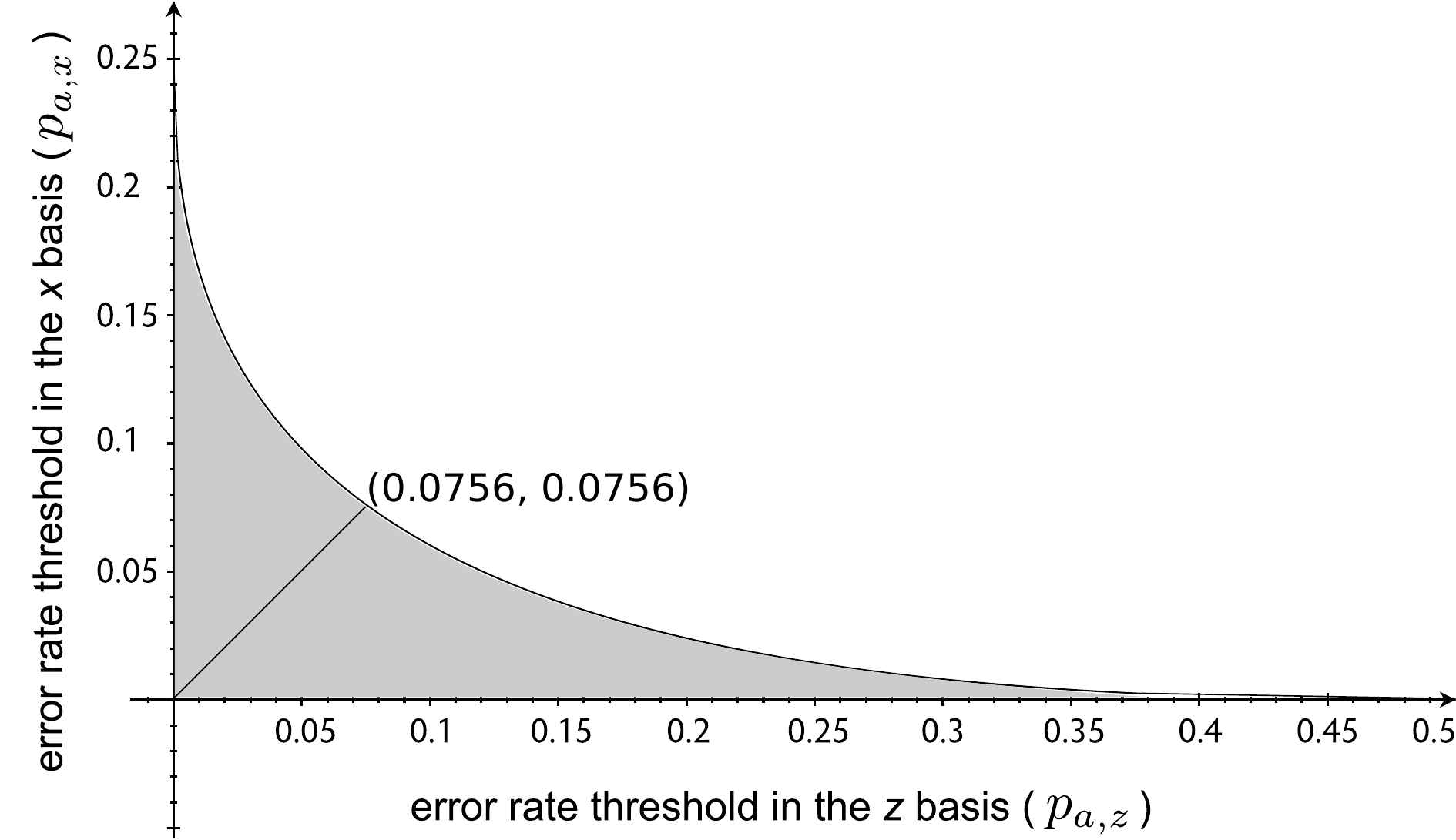}
\caption{\label{fig:security}\textbf{The secure asymptotic error rates
zone for BB84-INFO-$z$} (below the curve)}
\end{figure}

\section{\label{conclusion}Discussion}
In the current paper, we have proved the BB84-INFO-$z$ protocol
to be fully secure against collective attacks.
We have discovered that the results of BB84 hold very similarly
for BB84-INFO-$z$, with only two exceptions:
\begin{enumerate}
\item The error rates must be \underline{separately} checked to be
below the thresholds $p_{a,z}$ and $p_{a,x}$ for the TEST-Z and
TEST-X bits, respectively, while in BB84 the
error rate threshold $p_a$ applies to all the TEST bits
together.
\item The exponents of Eve's information (security) and of the
failure probability of the error-correcting code (reliability)
are different than in~\cite{BGM09}, because different numbers
of test bits are now allowed ($n_z$ and $n_x$ are arbitrary).
This implies that the exponents may decrease more slowly
(or more quickly) as a function of $n$.
However, if we choose $n_z = n_x = n$ (thus sending $N = 3n$
qubits from Alice to Bob), then we get exactly the same exponents
as in~\cite{BGM09}.
\end{enumerate}

The asymptotic error rate thresholds found in this paper
allow us to tolerate a higher
threshold for a specific basis (say, the $x$ basis) if we demand a
lower threshold for the other basis ($z$). If we choose the same
error rate threshold for both bases, then the asymptotic bound is
7.56\%, exactly the bound found for BB84
in~\cite{BBBMR06} and~\cite{BGM09}.

We conclude that even if we change the BB84 protocol to have INFO bits
only in the $z$ basis, this does not harm its security and reliability
(at least against collective attacks).
This does not even change the asymptotic error rate threshold.
The only drawbacks of this change are the need to check the
error rate for the two bases separately,
and the need to either send more qubits
($3n$ qubits in total, rather than $2n$)
or get a slower exponential decrease of the exponents required
for security and reliability.

We thus find that the feature of BB84, that both bases are used for
information, is not very important for security and reliability,
and that BB84-INFO-$z$ (that lacks this feature)
is almost as useful as BB84.
This may have important implications on the security and reliability of
other protocols that, too, use only one basis for information qubits,
such as~\cite{Mor98} and some two-way protocols~\cite{cbob07,calice09}.

We also present a better approach for the proof, that uses the quantum
distance between two states rather than the classical information.
In~\cite{BGM09},~\cite{BBBGM02}, and~\cite{BBBMR06},
the classical mutual information
between Eve's information (after an optimal measurement)
and the final key was calculated
(by using the trace distance between two quantum states);
although we should note that in~\cite{BGM09} and~\cite{BBBMR06},
the trace distance
was used for the proof of security of a single bit of the final key
even when all other bits are given to Eve, and only the last stages
of the proof discussed bounding the classical mutual information.
In the current paper, on the other hand, we use the trace distance
between the two quantum states until the end of the proof,
which allows us to prove fully composable security.

Therefore, our proof shows the fully composable security of
BB84-INFO-$z$ against collective attacks
(and, in particular, security even if Eve keeps her quantum states
until she gets more information when Alice and Bob use the key,
rather than measuring them at the end of the protocol);
and a very similar approach can be applied to~\cite{BGM09},
immediately proving the composable security of BB84
against collective attacks.
Our proof also makes a step towards making the security proof
in~\cite{BBBMR06} (security proof of BB84 against joint attacks)
prove the \emph{composable} security of BB84 against joint attacks.

Our results show that the BB84-INFO-$z$ protocol can be securely used
for distributing a secret key;
the security is of an ideal implementation
and against an adversary limited to collective attacks
(it may be possible to generalize the proof, so that it applies to
the most general attacks (joint attacks),
by using the methods of~\cite{BBBMR06},
but such generalization is beyond the scope of the current paper).
Moreover, the security of the final key is universally composable,
which means that the key may be used for any cryptographic purpose
without harming the security, even if Eve keeps her quantum states
and uses all the information she gets in the future in an optimal way.

The techniques described in our proof may be applied in the future
for proving the security of other protocols by using similar methods,
and, in particular, for proving the security of other QKD protocols
that use only one basis for the information bits,
such as~\cite{Mor98,cbob07,calice09} mentioned above.

We note that this paper strengthens the security proofs
described in~\cite{BGM09,BBBGM02,BBBMR06}, both because it slightly
generalizes them (from security of BB84 to security of BB84-INFO-$z$)
and because it makes them composable.
Those security proofs have various advantages over other methods to
prove security: first of all, they are mostly self-contained,
while other security proofs require many results from other areas
of quantum information (such as various notions of entropy needed for
the security proof of~\cite{renner_thesis08,bb84_sec_renner},
and entanglement purification and quantum error correction needed for
the security proof of~\cite{bb84_sec_SP});
second, they give tight finite-key bounds,
unlike several other methods (see details below);
and finally, at least in some sense, they are simpler than other
proof techniques. On the other hand, their generality and their
asymptotic error-rate threshold (7.56\%, rather than 11\%
given by~\cite{bb84_sec_renner,bb84_sec_SP}) are yet to be improved
by future research.

Our method for proving security gives explicit and tight
finite-key bounds. In contrast to this, the security proof
of~\cite{bb84_sec_SP} gives only asymptotic results
(for infinitely long keys).
For the security proof of~\cite{renner_thesis08,bb84_sec_renner},
it is proved today that for some protocols (including BB84),
one can get tight finite-key bounds~\cite{TLGR12}
that are the same as the ones found by our method;
but at first that security method gave very pessimistic bounds
(by using the de Finetti theorem~\cite{renner_thesis08,Ren07}),
and later, the bounds were improved for several protocols
(including BB84)~\cite{SR08}, but were still not tight
(see~\cite{TLGR12} for comparison).

We also note that the existence of many different proof techniques
is important, because some proofs may be more adjustable to
various QKD protocols or to practical scenarios;
some proofs may be clearer to different readers
with different backgrounds;
analyzing the differences between the proofs
and between their obtained results may lead to important insights
on the strengths and weaknesses of various techniques;
and the existence of many proofs makes the security result
more certain and less prone to errors.

We note that our security proof,
similarly to many other full security proofs of QKD,
assumes an ideal implementation (of ideal quantum systems consisting
of exactly one qubit) and theoretical attacks.
Practical implementations of QKD, almost always using photons, exist
(see~\cite{security_lo14,sec_practical09} for details);
their security analysis is much more complicated, because both Alice's
photon source and Bob's detector devices have weaknesses and deviations
from the theoretical protocol (especially when more than one photon
is emitted by Alice or is sent by the eavesdropper).
Those imperfections give rise to various practical attacks,
such as the ``Photon-Number Splitting'' attack~\cite{BLMS00}
(in which the eavesdropper takes advantage of emissions of
two or more photons by Alice and gets full information)
and the ``Bright Illumination'' attack~\cite{makarov10}
(in which the eavesdropper takes advantage of a weakness
of specific detectors used by Bob and gets full information).

Possible solutions to those problems of actual physical realizations
(see~\cite{security_lo14,sec_practical09} for more details)
include a much more careful analysis
of the practical devices and of practical implementations;
``Measurement-Device Independent''
QKD protocols~\cite{mdi_qkd1,mdi_qkd2,mdi_qkd3,mdi_qkd4},
that may be secure even if the measurement devices are controlled
by the adversary; and ``Device Independent''
QKD protocols~\cite{di_qkd1,di_qkd2,di_qkd3},
that may be secure even if all the quantum devices are controlled
by the adversary (under certain assumptions).

\section*{Acknowledgments}
The work of TM and RL was partly supported
by the Israeli MOD Research and Technology Unit.

\bibliographystyle{elsarticle-num}
\bibliography{security}

\end{document}